%% file: TvGP.tex
\documentclass{article}
\usepackage{arxiv}
\usepackage[hidelinks]{hyperref}
\usepackage{amsthm}
\newtheorem{lemma}{Lemma}

\input{Definitions}

\title{Tensor-variate Gaussian process regression for efficient emulation of complex systems: comparing regressor and covariance structures in outer product and parallel partial emulators.\thanks{Submitted to the editors 2024-12-12. This work was supported by the Chemical and Biological Technologies Department (contract HDTRA1-17-C-0028) and the UK Engineering and Physical Sciences Research Council (grants EP/R009902/1 and EP/S019472/1).}}

\begin{document}
\author{D. Semochkina\thanks{University of Southampton, UK 
  (\email{d.semochkina@soton.ac.uk}, \email{d.woods@soton.ac.uk}).}
\and S. E. Jackson\thanks{Durham University, Durham, UK (\email{samuel.e.jackson@durham.ac.uk})}
\and D. C. Woods\footnotemark[2]}
\maketitle

\begin{abstract}

Multi-output Gaussian process regression has become an important tool in uncertainty quantification, for building emulators of computationally expensive simulators, and other areas such as multi-task machine learning. We present a holistic development of tensor-variate Gaussian process (TvGP) regression, appropriate for arbitrary dimensional outputs where a Kronecker product structure is appropriate for the covariance. We show how two common approaches to problems with two-dimensional output, outer product emulators (OPE) and parallel partial emulators (PPE), are special cases of TvGP regression and hence can be extended to higher output dimensions. Focusing on the important special case of matrix output, we investigate the relative performance of these two approaches. The key distinction is the additional dependence structure assumed by the OPE, and we demonstrate when this is advantageous through two case studies, including application to a spatial-temporal influenza simulator. 
\end{abstract}



\section{Introduction}\label{sec:intro}

Scientific processes are commonly modeled using mathematical models implemented in computer codes, or simulators, that encapsulate the key features of the system and facilitate prediction and decision-making \cite{calder2018}.
When these simulators are computationally expensive, it is common to approximate them using statistical emulators, most commonly a Gaussian process, constructed from computer experiments \cite{santner2003tda}. When the simulator output is dynamic, evolving over dimensions such as space or time, development of efficient and effective emulators is more challenging.

As motivation, we study a simulator of the spatio-temporal dynamics of an influenza epidemic, developed by the UK Health Security Agency (UKHSA). The code implements an SEIR model \cite{hethcote2000} that describes the time progression of a closed population through the states \textit{Susceptible}, \textit{Exposed}, \textit{Infected} and \textit{Recovered}. The model also has a discrete spatial dimension, made up of a set of ``patches'' representing different urban communities. Movement between these patches is controlled via a commuter matrix which simulates flows of people for work, education and leisure. While there is a well-understood, although intractable, set of differential equations governing this system, for the purposes of this paper we treat the simulator as a block box modeling the number of infected individuals and each simulator run having as output a matrix of spatio-temporal responses (numbers of infections for each time and patch).   

Previous work \cite{higdon2008cmc,conti2010beo,bowman2016eom, overstall2016meo} on emulating simulators with dynamic or, relatedly, multivariate output has focused on vector-valued responses, where the output from each run is unstructured or the structure is preserved in, or collapsed onto, a single dimension. Most commonly, the covariance kernel for the Gaussian process emulator is formed as the product of two separate kernels describing, respectively, relationships between simulator inputs and simulator output indices.  Two approaches which have seen particular application are (i) the Outer Product Emulator \cite{rougier2008eef} (OPE), which extends the separability to include the mean function with a closely related model applied to spatio-temporal simulators \cite{bzkl2013}; and (ii) the Parallel Partial Emulator \cite{gu2016ppg} (PPE), which assumes separate emulators across the dynamic output indices with common regression functions, but distinct regression parameters, and common kernel functions with shared parameters.

Motivated by work in the machine learning literature \cite{xu2012,chu2009, zhao2014, campbell2020}, in this paper, we cast OPE and PPE as special cases of a tensor-variate Gaussian process (TvGP) which allows modeling of outputs of arbitrary dimension represented as tensors using separable covariance kernels (Section~\ref{sec:emulation}). We then perform an empirical comparison of the effectiveness of OPE and PPE via two case studies: the SEIR patch model described above (Section~\ref{sec:app1}) and a synthetic environmental simulator \cite{bliznyuk2008bayesian} of a pollutant spillage from a chemical accident (Section~\ref{sec:app2}). A primary aim of these studies is to assess the contribution of the different assumptions made for the two emulators, and to attempt to gain insight about situations when imposing beliefs about correlation structure is important.  As such, a discussion of the results and the differences between the two emulators are also provided (Section~\ref{sec:conc} and Appendix~\ref{appendix:comp}). Code for implementing the methods in this paper and reproducing the results is available at \url{https://github.com/samjacksonstats/TVGP}.

\section{Multivariate Emulation}
\label{sec:emulation}

\subsection{Separable emulators}\label{subsec:sep}

\nom{\bof}{simulator output vector,matrix or tensor}
\nom{\bx}{generic vector input to simulator}
\nom{\mX}{input domain}
\nom{p}{simulator input dimension}
\nom{r}{overall simulator output dimension, equal to $r_1 \times r_2$ later on}
\nom{\bmu}{emulator's mean}
\nom{\ex}{$r$-dimensional stochastic residual process}

Consider a generic multivariate simulator, represented by the function $ \bof $, which takes a vector $ \bx \in \mX \subset \real^p $ of input parameters, and outputs a vector $ \fx \in \real^r $.  We choose to represent our beliefs about the behavior of the simulator output for any input $ \bx $ in the following form \cite{goldstein2004pff}:
\be\label{eq:emulator}
\fx = \bmu(\bx) + \ex \,, 
\ee
where 
$\bmu : \real^p \rightarrow \real^{r}$ is a vector-valued mean function, typically expressed as a linear predictor, and $ \ex $ is an  $ r $-dimensional stochastic residual process aiming to capture residual variation across $ \mX $.

\nom{\kappa(\bx,\bx';\theta)}{correlation function, parameterised by $\theta$}
\nom{\theta}{parameters of correlation function $\kappa$}
\nom{\Sig}{output variance matrix}

A prior belief specification across the collection $ \{ \bmu(\bx),  \ex : \bx \in \mX \} $ leads to a prior specification about simulator output $ \fx $ at any $ \bx \in \mX $.
Such specification is not trivial; however, it is common for the specification to assume $ \e{\ex} = 0 $, $ \cov{\bmu(\bx)}{\ex} = \mathbf{0} $, and a separable covariance structure between the value of $ \ex $ and $ \epr $, at two inputs $ \bx $ and $ \bx' $, of the following form \cite{conti2010beo}:
\be
\cov{\ex}{\epr} =  \kappa(\bx,\bx'\,; \boldsymbol{\theta}_0) \, \Sig \label{eq:cov_fct} \, .
\ee
This assumption has been shown to be very effective for constructing emulators for multivariate simulators with large $r$. Covariance between the output components at any given input is given by the product of the positive semi-definite coregionalization matrix $ \Sig \in \real^{r \times r} $ and
the stationary kernel function $\kappa : \real^p\times \real^p \rightarrow \real$, parameterised by $\boldsymbol{\theta}_0$ (although we will drop the $\boldsymbol{\theta}_0$ in this expression from here-on in, assuming that such a parameterization is implicit). A common choice is the Gaussian correlation function \cite{kennedy2001bco, bayarri2007aff}:
\be\label{eq:gcf}
\kappa(\bx,\bx') = \exp \left ( - \summ{h}{1}{p} \left \{ \frac{x_h -x'_h}{\theta_{0,h}} \right \} ^2 \right ),
\ee
which requires specification of correlation length parameters $ \theta_{0,h}, h=\seq{1}{p} $.
%

One particular specification of a prior stochastic process for $\bof(\bx)$ with these properties is a multivariate Gaussian process (MvGP) \cite{conti2010beo,overstall2016meo},
\be
\label{eq:mgp}
\bof^T(\cdot) \sim \MGP( \bmu^\top(\cdot), \, \kappa(\cdot, \cdot), \Sig)\,.
\ee
This process is a multivariate generalisation for $r$-dimensional outputs of the usual one-dimensional Gaussian process \cite{rasmussen2006gpf}, and captures many of the representations developed throughout the literature.
\nom{\MGP}{multivariate Gaussian process distribution}
\subsection{Tensor-variate Gaussian processes}\label{subsec:tvgp}

There are several commonly observed structures in $\fx$, for example with entries corresponding to a spatial location in two dimensions or defined via space-time coordinates in higher dimensions. In Sections~\ref{sec:app1} and~\ref{sec:app2} we emulate simulators with two-dimensional space-time coordinates, that is, input $\bx$ produces outputs
\begin{equation}
\ten{F}(\bx)=
\begin{pmatrix}
 f_{11}(\bx)& \ldots & f_{1 r_2}(\bx)\\
 \vdots & \ddots & \vdots \\
 f_{r_1 1}(\bx) & \ldots & f_{r_1 r_2}(\bx)
 \end{pmatrix}.\nonumber
\end{equation}
Here $f_{i_1 i_2}(\bx) = f(\bx; s_{i_1}, t_{i_2})$ represents model output for spatial grid point $s_{i_1} \in \mS, i_1 = 1,\ldots,r_1$ at time $t_{i_2} \in \mT, i_2 = 1,\ldots,r_2$ corresponding to the model input parameters $\bx$. Sets $\mS$ and $\mT$ can be thought of as discretised sets of points within continuous spatial and temporal domains $\mathbb{S}$ and $\mathbb{T}$ respectively. That is, $\mS$ is a set of coordinate points of interest over a geographical region represented by $\mathbb{S}$, and $\mT$ is a set of time increments (hours or days) over time $\mathbb{T}$. 
More generally, if the output arises on a regular grid of higher-order coordinates, the simulator output can be represented by an $m$-dimensional array or tensor $\ten{F}(\bx) \coloneqq f_{i_1\ldots i_m}(\bx)\,:\, \times_{z=1}^m r_z$, with $i$th dimension having size $r_z$ and $\br^{\top} = (r_1,\ldots,r_m)$.
\nom{\ten{F}(\bx)}{$m$-dimensional array or tensor of simulator outputs at input $\bx$}
\nom{i_1}{indexing spatial output $\seq{1}{r_1}$}
\nom{i_2}{indexing temporal output $\seq{1}{r_2}$}
\nom{r_1}{number of spatial output components}
\nom{r_2}{number of temporal output components}
\nom{m}{tensor dimensionality}
\nom{\mS}{discrete set of spatial points in domain $\mathbb{S}$}
\nom{\mT}{discrete set of temporal points in domain $\mathbb{T}$}
\nom{\mathbb{S}}{(potentially) continuous spatial domain}
\nom{\mathbb{T}}{(potentially) continuous temporal domain}

Such output can be modeled using~\eqref{eq:mgp} via a vectorization operation \cite{ohlson2013} 
\be
\label{eq:vec}
\fx = \vect\left({\ten{F}(\bx)}\right) = \sum_{\mathbf{i}_\br}f_{i_1\ldots i_m}(\bx) \, \bm{\varepsilon}^\br_{i_1:i_m}\,,
\ee
where $\mathbf{i}_\br = \{(i_1,\ldots,i_m)\,:\, 1\le i_z \le r_z,\, 1\le z\le m\}$, $\bm{\varepsilon}^\br_{i_1:i_m} = \bm{\varepsilon}^{r_1}_{i_1}\otimes \cdots \otimes \bm{\varepsilon}^{r_m}_{i_m}$ and $\bm{\varepsilon}^r_i$ is a $r$-dimensional unit vector with $i$th entry 1. However, the coregionalization matrix would have dimension $r = \prod_{z=1}^m r_z$, possibly making inference and prediction from~\eqref{eq:mgp} computationally infeasible except in special cases such as output independence and diagonal $\Sig$, as utilized by the PPE. 

Alternatively, if an assumption of separability between the output dimensions is reasonable, the coregionalization matrix can be decomposed as
\nom{\bm{\varepsilon}^r_i}{an $r$-dimensional unit vector with $i$th entry equal to 1 and all other equal to 0}
\nom{\mathbf{i}_\br}{set of all possible combinations of all tensor indices $\{(i_1,\ldots,i_m)\}$}
\be
\label{eq:sig-tp}
\Sig = \Sig_1\otimes \cdots \otimes\Sig_m\,,
\ee
leading to a \textit{tensor-variate Gaussian process} (TvGP)
\nom{\Sig_i}{output variance matrix for output $i$}
\be
\label{eq:tgp}
\ten{F}(\cdot) \sim \TGP( \ten{M}(\cdot), \, \kappa(\cdot, \cdot), \Sig_1,\ldots,\Sig_m)\,,
\ee
with $\ten{M}\,:\, \real^p\rightarrow \real^{r_1\times \cdots \times r_m}$  a $m$-dimensional tensor-variate function. Such a process can be defined via a linear combination of independent univariate standard normal random variables, $\bu^{\rm T} = (u_1,\ldots,u_r)$, $u_l\sim \mathcal{N}(0, 1)$, via  
\nom{\ten{M}(\cdot)}{$m$-dimensional tensor-variate Gaussian process mean function}
\nom{\bu}{vector of independent univariate standard normal random variables}
$$
\vect\left({\ten{F}(\bx)}\right) = \vect\left({\ten{M}(\bx)}\right) + \kappa(\cdot, \cdot)^{\frac{1}{2}}\Sig^{\frac{1}{2}}\bu\,,
$$
where $\Sig$ has the Kronecker product form from~\eqref{eq:sig-tp} and any matrix square root can be used. In an analogy to \eqref{eq:emulator}, such a process can also be expressed as
\be
\label{eq:add_tvgp}
\ten{F}(\bx) = \ten{M}(\bx) + \ten{E}(\bx)\,,
\ee
where $\ten{E} \coloneqq e_{i_1\ldots i_m}(\bx)\,:\, \times_{z=1}^mr_z$ and $\ten{E}\sim \TGP( \ten{Z}, \, \kappa(\cdot, \cdot), \Sig_1,\ldots,\Sig_m)$, with $\ten{Z}$ the tensor with all components equal to 0.

A consequence of assuming this prior is that, for a collection of $n$ input vectors $\bx^{(1)},\ldots,\bx^{(n)}$, the $n\times r_1 \times \cdots \times r_m$ output tensor $\ten{F} = \{\ten{F}(\bx^{(1)}),\ldots, \ten{F}(\bx^{(n)})\}$ follows a tensor, or multi-linear, normal distribution \cite{ohlson2013} of dimension $m+1$
\nom{n}{number of training points}
\be\nonumber
\ten{F} \sim \TN_{n, r_1,\ldots, r_m}(\ten{M}, \bK, \Sig_1,\ldots, \Sig_m)\,,
\ee
with $m+1$ dimensional mean tensor $\ten{M} = \{\ten{M}(\bx^{(1)}), \ldots, \ten{M}(\bx^{(n)})\}$ and $\bK$ having $jk$-th entry $\kappa(\bx^{(j)}, \bx^{(k)})$. Vectorization of $\ten{F}$ results in a multivariate normal distribution
\nom{\bK}{input variance matrix}
\nom{\ten{M}}{$m+1$-dimensional mean tensor of $\ten{M}(\cdot)$, evaluated at $\seq{\bx^{(1)}}{\bx^{(n)}}$}
\be\label{eq:vecMN}
\vect\left(\ten{F}\right)\sim \mathcal{N}_{nr}(\vect(\ten{M}), \bK\otimes (\otimes_{i=1}^m\Sig))\,.
\ee

TvGP regression is defined via the joint distribution of $\ten{F}$ and $\tilde{\ten{F}} = \{\ten{F}(\tilde{\bx}^{(1)}), \ldots, \ten{F}(\tilde{\bx}^{(n^{\prime})})\}$ at new inputs $\tilde{\bx}^{(1)}, \ldots, \tilde{\bx}^{(n^{\prime})}$, given by
\be\nonumber
\{\ten{F}, \tilde{\ten{F}}\} \sim \mathcal{TN}_{n+n^{\prime}, r_1, \ldots, r_m}\left(\ten{M}_+, \bK_+, \Sig_1, \ldots, \Sig_m\right)\,,
\ee
with $\ten{M}_+ = \{\ten{M}, \tilde{\ten{M}}\}$, $\tilde{\ten{M}} = \{M(\tilde{\bx}^{(1)}), \ldots,  M(\tilde{\bx}^{(n^{\prime})})\}$,
\nom{\tilde{\bx}}{generic new vector input to simulator (as opposed to $\bx^{(1)},\ldots,\bx^{(n)}$)}
\nom{n^{\prime}}{number of additional training points}
\nom{\tilde{\ten{M}}}{$m+1$-dimensional mean tensor of $\ten{M}(\cdot)$, evaluated at $\seq{\tilde{\bx}^{(1)}}{\tilde{\bx}^{(n^{\prime)}}}$}
\nom{\ten{M}_+}{$\{\ten{M}, \tilde{\ten{M}}\}$}
\be\nonumber
\bK_+ = 
\begin{bmatrix}
\bK & \bL \\
\bL^T & \tilde{\bK}
\end{bmatrix}\,,
\ee
$\tilde{\bK}_{j^{\prime}k^{\prime}} = \kappa(\tilde{\bx}^{(j^{\prime})}, \tilde{\bx}^{(k^{\prime})})$ ($j^{\prime},k^{\prime} = 1,\ldots, n^{\prime}$) and $\bL_{jk^{\prime}} = \kappa(\bx^{(j)}, \tilde{\bx}^{(k^{\prime})})$ ($j = 1,\ldots, n;\,k^{\prime} = 1,\ldots,n^{\prime}$). Then, conditional on $\ten{F}$ we have the posterior predictive distribution
\nom{\bL}{old input $\bx$ and new input $\tilde{\bx}$ covariance matrix}
\nom{\tilde{\bK}}{input variance matrix at new inputs $\tilde{\bx}^{(1)}, \ldots, \tilde{\bx}^{(n^{\prime})}$}
\be\nonumber
\tilde{\ten{F}} | \ten{F} \sim \mathcal{TN}_{n^{\prime}, r_1,\ldots, r_m}\left(\tilde{\ten{M}}^\star, \tilde{\bK}^\star, \Sig_1, \ldots, \Sig_m \right)\,,
\ee
with $\tilde{\ten{M}}^\star = \{\ten{M}^\star(\tilde{\bx}^{(1)}), \ldots, \ten{M}^\star(\tilde{\bx}^{(n^{\prime})})\}$,
\be\label{eq:TvGPpredmean}
\ten{M}^\star(\tilde{\bx}) = \ten{M}(\tilde{\bx}) + [\ten{F} - \ten{M}]\times_1\kappa(\tilde{\bx})^T\bK^{-1}\,,
\ee
$\kappa(\tilde{\bx}) = (\seq{\kappa(\tilde{\bx},\bx^{(1)})}{\kappa(\tilde{\bx},\bx^{(n)})})^T$, $\tilde{K}^\star_{j'k'} = \kappa^\star(\tilde{\bx}^{(j')}, \tilde{\bx}^{(k')})$ and 
\be\nonumber
\kappa^\star(\tilde{\bx}, \tilde{\bx}^\prime) = \kappa(\tilde{\bx}, \tilde{\bx}^\prime) - \kappa(\tilde{\bx})^T\bK^{-1}\kappa(\tilde{\bx}^\prime)\,.
\ee
Here, $\times_c$ is the $c$-mode tensor-matrix product defined for tensor $\ten{P}$ and matrix $\Lambda$, with entries $p_{i_1 \cdots i_m}$ and $\lambda_{ab}$ respectively, as
\nom{\tilde{\bK}^\star}{posterior input covariance matrix}
\nom{\tilde{\ten{M}}^\star}{posterior $m+1$-dimensional mean tensor}
\nom{\ten{P}}{generic tensor}
\nom{\Lambda}{generic matrix}
\nom{\times_c}{$c$-mode tensor-matrix product}
\nom{\kappa^\star}{updated (posterior) covariance function}
$$
(\ten{P}\times_c\Lambda)_{i_1\ldots i_{c-1} a i_{c+1}\ldots i_m} = \sum_b \lambda_{ab}p_{i_1\ldots i_{c-1} b i_{c+1}\ldots i_m}\,.
$$
Hence we have the posterior process
\be\nonumber
\ten{F}(\bx) | \ten{F} \sim \mathcal{TVGP}(\ten{M}^\star(\cdot), \kappa^\star(\cdot, \cdot), \Sig_1, \ldots, \Sig_m)\,.
\ee

The statistical surrogate model of $\bof(\cdot)$ is usually not specified with fixed chosen values for the elements of $\Sig$ and $\theta$. If a prior distributional specification is assumed across this collection, the updated posterior distributional form $\ten{F}$ is no longer a Gaussian process, and numerical inference methods must be used. Typically, a vectorization of the mean tensor $\ten{M}(\cdot)$ is formed as a combination of a linear predictor and unknown parameters (Appendix~\ref{appendix:comp}). These mean parameters are integrated over with respect to a prior distribution, retaining a posterior predictive Gaussian process. Covariance parameters are often estimated as, e.g., posterior modes and plugged into the posterior predictive density \cite[ch. 5]{gramacy2020surrogates}.  


\subsection{Emulation of Dynamic Simulators}\label{sec:EmulationOfDynamicSimulators}

A straightforward approach for emulation of spatio-temporal simulators is to treat the $ r = r_1 \times r_2 $ space-time outputs
as a general multivariate collection of random variables and emulate directly using the MvGP from~\eqref{eq:mgp} \cite{ohagan1999uaa, conti2010beo}. In this case, the indexing of the model output enters only indirectly through specification of the prior variance matrix for the outputs. Such an emulator is simple to implement, however, computationally expensive for large values of $r$. It may also be difficult to specify or elicit the required parameters involved in capturing the covariances between all of the outputs.  

An alternative approach is to construct a separate emulator for each space-time output, indexed by $ i = \seq{1}{r} $ \cite{conti2010beo}:
\be
f_i(\bx) = \mu_i(\bx) + e_i(\bx)  \label{sep_em},
\ee
with $e_i$ given independent, zero-mean, Gaussian process priors. As a result, only data for the $ i $th output is utilised when emulating output at spatial-temporal coordinate $i$. Prior correlation across space-time output is not captured. Additionally, constructing each emulator completely independently, including specification of parameters, is computationally intensive; commonly, $r_1$ and/or $r_2$ may be in the order of (tens of) thousands; hence, even if using parallel computation,  estimation will be expensive. Even for smaller grid sizes, computation can be reduced by orders of magnitude through the use of a TvGP.
\nom{\bb_i}{vector of regression parameters for output $i$}
\nom{\bg(\cdot)}{regressors function}

Extensions or variants of these two approaches have proved popular in applications \cite{rgmr2009,sbbcppw2014,kitnya2023,Gao_2024}. In the next sections, we place two such methods in the context of tensor-variate Gaussian processes with $m=2$ dimensional tensor (or matrix) output $\ten{F} = f_{i_1i_2}(\bx)\,:\,r_1\times r_2$. 


\subsection{Outer Product Emulators}\label{sec:OPE}

An Outer Product Emulator (OPE) for multivariate responses \cite{rougier2008eef} extends the separable emulator~(\ref{eq:mgp}) by assuming a mean function constructed from a further separable structure in the regressors. Such structure can be extended to higher-dimensional outputs,  defined over the inputs and all output locations within the TvGP expressed as \eqref{eq:add_tvgp}. Each component is modeled as

$$
\simx{i_1 \cdots i_m}=\bg_{i_1 \cdots i_m}^\top(\bx)\bb +e_{i_1 \cdots i_m}(\bx),\quad i_z=1,\ldots,r_z;\ z=1,\ldots,m,
$$
where $\bg_{i_1 \cdots i_m}^{\top}(\bx)=\bg_{0}^{\top}(\bx) \otimes (\otimes_{z=1}^m\bg_{z}^{\top}(\bxi_{zi_z})),\ \boldsymbol{\beta}=\boldsymbol{\beta}_{0} \otimes (\otimes_{z=1}^m\boldsymbol{\beta}_{z})$, with $\boldsymbol{g}_0(\cdot)$ and $\boldsymbol{\beta}_0$ being $v_0$-vectors of regressors and parameters related to the input variables, and $\boldsymbol{g}_z(\cdot)$ and $\boldsymbol{\beta}_z$ $v_z$-vectors corresponding to the output locations in dimension $z$. Hence, the TvGP has mean tensor
\be\label{eq:OPEmean}
\ten{M}(\bx) = \mathbf{g}^\top_{i_1 \cdots i_m}(\bx)\bb\,:\,\times_{z=1}^mr_z\,.
\ee

For the $m=2$ dimensional spatio-temporal case, we define $\bxi_{1i_1} = s_{i_1}\in\mathcal{S}$ as the spatial locations and temporal locations $\bxi_{2i_2} = t_{i_2}\in\mathcal{T}$ and 
$$
\begin{aligned}
\mathcal{M}(\bx) & =
\begin{bmatrix}
\bg_{11}^{\top}(\bx) \boldsymbol{\beta} & \cdots & \bg_{1 r_{2}}^{\top}(\bx) \boldsymbol{\beta} \\
\vdots & \ddots & \vdots \\
\bg_{r_{1} 1}^{\top}(\bx) \boldsymbol{\beta} & \cdots & \bg_{r_{1} r_{2}}^{\top}(\bx) \boldsymbol{\beta}
\end{bmatrix}\\
& =G(\bx)\left[I_{r 2} \otimes \boldsymbol{\beta}\right]
\end{aligned},
$$
where 
$$G(\bx)=
\begin{bmatrix}
\bg_{11}^{\top}(\bx) & \cdots & \bg_{1 r_{2}}^{\top}(\bx)  \\
\vdots & \ddots & \vdots \\
\bg_{r_{1} 1}^{\top}(\bx) & \cdots & \bg_{r_{1} r_{2}}^{\top}(\bx)
\end{bmatrix}\,,
$$
and $I_{r_2}$ is an $r_2\times r_2$ identity matrix. 

The TvGP is completed by specification of coregionalization matrix $\Sig = \otimes_{z=1}^m \Sig_z$ for residual tensor $\ten{E}$. For emulators with spatio-temporal output, 
\be
\Sig=\Sig_{1} \otimes \Sig_{2}\,,\nonumber
\ee
with $\Sig_{1}$ an $r_1\times r_1$ spatial coregionalization matrix and $\Sig_{2}$ an $r_2\times r_2$ temporal coregionalization matrix. Advantage can be taken of naturally defined metric spaces on the output locations to borrow strength across different space-time coordinates via suitable prior specifications. Here, we assume $\Sig = \sigma^2W_1\otimes W_2$. Hence $\Sig_{1} = W_1$, with $r_1\times r_1$ matrix $W^s$ having entries defined via a kernel specifying response correlation on the spatial output locations, and $\Sig_{2}=\sigma^{2} W_2$, where $r_2\times r_2$ matrix $W^t$ is similarly defined via a kernel for response correlation on the temporal output locations. The single scale parameter is $\sigma^2$. Efficient computational methods are available for the OPE with matrix-variate output as appropriate for spatio-temporal simulators; see \cite{bzkl2013} and Appendix~\ref{appendix:vec}, which demonstrates that a vectorization of the tensor-variate OPE maintains the usual form of a vector-variate OPE.


\subsection{Parallel Partial Emulators}\label{sec:PPE}

A Parallel Partial Emulator (PPE), introduced in \cite{gu2016ppg}, involves fitting separate, shared-characteristic emulators at each output grid point: 
\be
\simx{i_1 \cdots i_m} = \bg_0^\top(\bx) \bb_{i_1 \cdots i_m} + e_{i_1 \cdots i_m}(\bx)\,.
\ee
The emulators have a common $v_0$-vector of basis functions $ \bg_0(\cdot)$, but distinct $v_0$-vectors of regression coefficients $\bb_{i_1 \cdots i_m}$. Emulators at different output locations are assumed conditionally independent, a priori, and hence $\Sig = \operatorname{diag}\left(\sigma^2_1, \ldots, \sigma^2_r\right)$. The PPE can be expressed as a TvGP with a single output dimension of size $r = \prod_{z=1}^m r_z$, that combines all output indices. This is equivalent to MvGP~(\ref{eq:mgp}). The emulator has mean vector 
\be\label{eq:PPEmean}
\ten{M}(\bx)= \left[I_{r} \otimes \bg_{0}^{\top}(\bx)\right]\vect\left(\ten{B}\right)\,,
\ee
with $\ten{B} = \beta_{i_1 \cdots i_mj}\,:\,[\times_{z=1}^m r_z]\times  v_0$. Here $\ten{B}$ is a tensor representation of the output-specific ($i_1 \cdots i_m$) coefficients of the common $v_0$ regressors and $\beta_{i_1 \cdots i_m j}$ is the $j$th element of $\bb_{i_1 \cdots i_m}$.

The shared structure of the linear predictor and kernel function, including common parameters $\boldsymbol{\theta}$ results in the PPE inheriting properties of the simulator, including smoothness \cite{gu2016ppg}. Completely independent emulators for each space-time coordinate would not have this property.  

\subsection{Priors and estimation}\label{sec:priorest}

Connections between the OPE and PPE can be established via the choice of hyperparameter priors \cite{rougier2008eef, conti2010beo, gu2016ppg}, i.e., for $\bb$ and $\Sigma_z$; see Appendix~\ref{appendix:coreg}. Here, for the OPE we use a conjugate normal-inverse gamma joint prior for $\bb,\sigma^2$ \cite{rougier2008eef}, leading to a marginal multivariate $t$-distribution prior for $\bb$

In the mean structure, for both emulators, we specify input regressors $\mathbf{g}_0^\top = (1, \bx^\top)$. For the OPE we also specify temporal regressors $\mathbf{g}_2^\top = (1, t)$. The input correlation for both emulators was specified as Gaussian from \eqref{eq:gcf}. For the OPE, the Gaussian correlation function was also used to parameterize the temporal correlations, with
$$
(\Sig_2)_{i_2i_2^\prime} = \sigma^2\exp\left(-\{t_{i_2} - t_{i_2^\prime}\}^2/\theta_2\right)\,,
$$
and hyperparameter $\theta_2>0$. The choice of spatial regressors and correlation function for the OPE varies between our two examples. 

The OPE is completed via choice of parameters for the normal-inverse gamma prior discussed above; we use $\sigma^{-2}\sim \text{Gamma}(1, 1)$ and $\bb|\sigma^2\sim N(\mathbf{0}, \sigma^2I_{v_0v_1v_2})$. Hence, the marginal prior distribution for $\bb$ is a scaled $t$-distribution with two degrees of freedom, and enforces shrinkage towards zero for the posterior of the regression parameters. All correlation parameters were estimated via maximum likelihood, which corresponds to minimizing the Mahalanobis distance, using the L-BFGS-B algorithm \cite{byrd1995limited}.

For the PPE, vague prior beliefs are assumed for $ \bb_{i_1i_2} $, leading to posterior mean and variance estimators for $ \bb_{i_1i_2} $ equivalent to generalized least squares estimators \cite{jackson2018dop}. We obtain point estimates for $ \sigma_i^2 $ and common correlation parameters $ \boldsymbol{\theta} $ via maximum likelihood \cite{andrianakis2012teo, jackson2020uhc}, with the latter obtained numerically.

Two different computer experiments were designed and performed to generate training data, with $n = 20$ ($\mX_{20}$) and $n=50$ ($\mX_{50}$). Both designs were found using the MaxPro \cite{joseph2015mpd} space-filling criterion, with input dimensions each scaled to $[-1,1]$. All statistical modeling also used this scaling, which was also applied to spatial and temporal locations. A further MaxPro design, $\mX_D$ of size $n^\prime = 150$, was constructed and run for diagnostic purposes.

\section{Application: Dynamic Influenza Simulator}
\label{sec:app1}

\subsection{The Simulator}

Our first application is to the compartmental SEIR simulator for influenza introduced in Section~\ref{sec:intro}. The spatial patches \cite{vdDriessche} represent population distributions among 15 fictional regions. The population in each patch moves through the four states ($S$, $E$, $I$ and $R$) over time according to a system of deterministic differential equations, governed by a set of rate parameters which are consistent across all patches. Previous work on emulating SEIR simulators has not included such a spatial component \cite{farah2014bea}.

Interactions between infection levels in neighboring patches are represented by a non-symmetric (commuter) matrix, through which movement between the different patches is modeled \cite{sd1995}, with row $a$ column $b$ representing the flow of people from patch $a$ to patch $b$. Two daily timesteps are assumed, simulating travel from a ``home" patch to a ``work" patch and back again. The susceptible population in patch $a$ can become infected if they mix with (i) infected individuals from patch $a$; (ii) infected individuals in patch $a$ who have moved from patch $b$; or (iii) infected individuals within patch $b$ if they themselves commute to patch $b$.

Here, spatio-temporal emulators are built for the number of infected individuals in each patch over time. The simulator is initialized with 100 infected individuals in a single fixed patch with the remaining population across all patches being susceptible. The input parameters to the simulator are rates $x_1=\beta\in [1, 2]$, $x_2=\alpha\in [0.25, 1]$ and $x_3=\gamma\in [0.25, 1]$, which are constant across space and time and effectively control the modeled transmission rates from $S$ to $E$, $E$ to $I$ and $I$ to $R$, respectively. Initial conditions for the simulator and the ranges for these inputs were chosen by the epidemiological modelers at UKHSA.

\subsection{Emulation}\label{sec:patch_emul}

Tensor-variate GP emulators were constructed for this simulator, under each of the OPE (Section~\ref{sec:OPE}) and PPE assumptions (Section~\ref{sec:PPE}). The simulator output at each space-time location, $y_{i_1i_2}(\mathbf{x})$, was shifted and log transformed before emulation; that is, $\simx{i_1 i_2} = \log\left(y_{i_1i_2}(\mathbf{x})+1\right)$. Each simulator run produced numbers of infected individuals across $r_1 = 15$ spatial patches and $r_2=150$ discrete time steps.

The OPE and PPE were constructed using the priors described in Section~\ref{sec:priorest}. The fixed initial conditions and commuter matrix led, for all combinations of input variables, to the same temporal order of patches obtaining their respective maximum number of infections. This ordering was used to define a new spatial coordinate for each patch, $\tilde{s}_{i_1}$, taking values in $\{1,\ldots,15\}$, which are common across all simulator runs. Hence these quantities can be used to define $\Sig_1$ for the OPE. This matrix was parameterized using a Gaussian correlation function and therefore for the OPE
$$
(\Sig_1)_{i_1i_1^\prime} = \exp\left(-\{\tilde{s}_{i_1} - \tilde{s}_{i_1^\prime}\}^2/\theta_1\right)\,,
$$
with hyperparameter $\theta_1>0$. For the prior mean structure, we use $\bg_{1}^{\top}(\mathring{s}_{i_1}) = (1, \mathring{s}_{i_1}),$ where $\mathring{s}_{i_1}$ is a scaled patch size. 

\subsection{Results}\label{sec:patchresults}

\begin{figure}
 \begin{center}
\includegraphics[width=\linewidth]{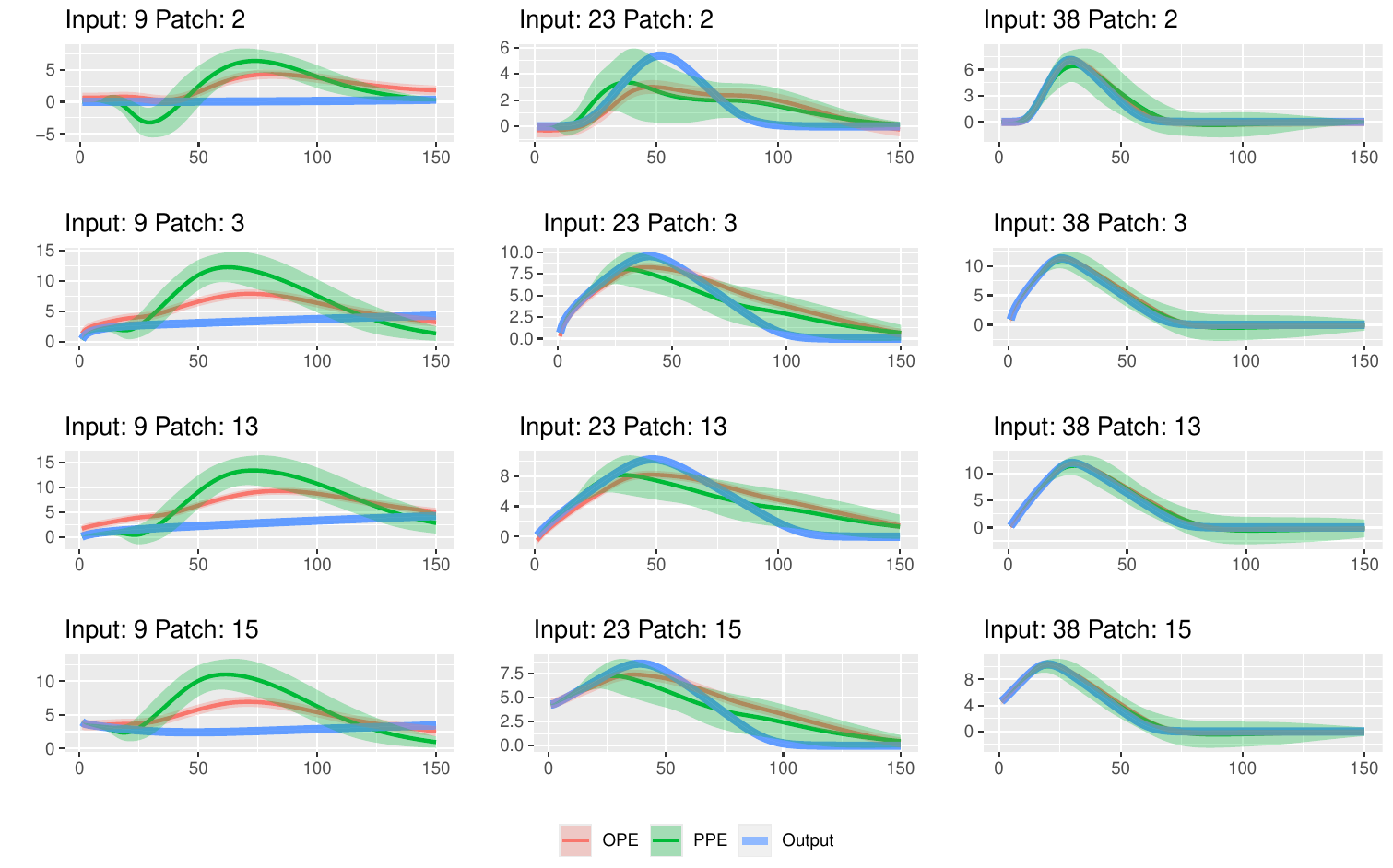}
\caption{Emulator mean predictions (solid lines) and 95\% credible intervals (shaded area) for the influenza simulator for three test input points (columns) and four patches (rows). The OPE (red) and PPE (green) were were built using $n=20$ training points. In each case, true simulator output is given as a blue line for comparison. Predictions are on the log scale.} 
\label{fig:OPEvPPE20}
\end{center}
\end{figure}

\begin{figure}
 \begin{center}
\includegraphics[width=\linewidth]{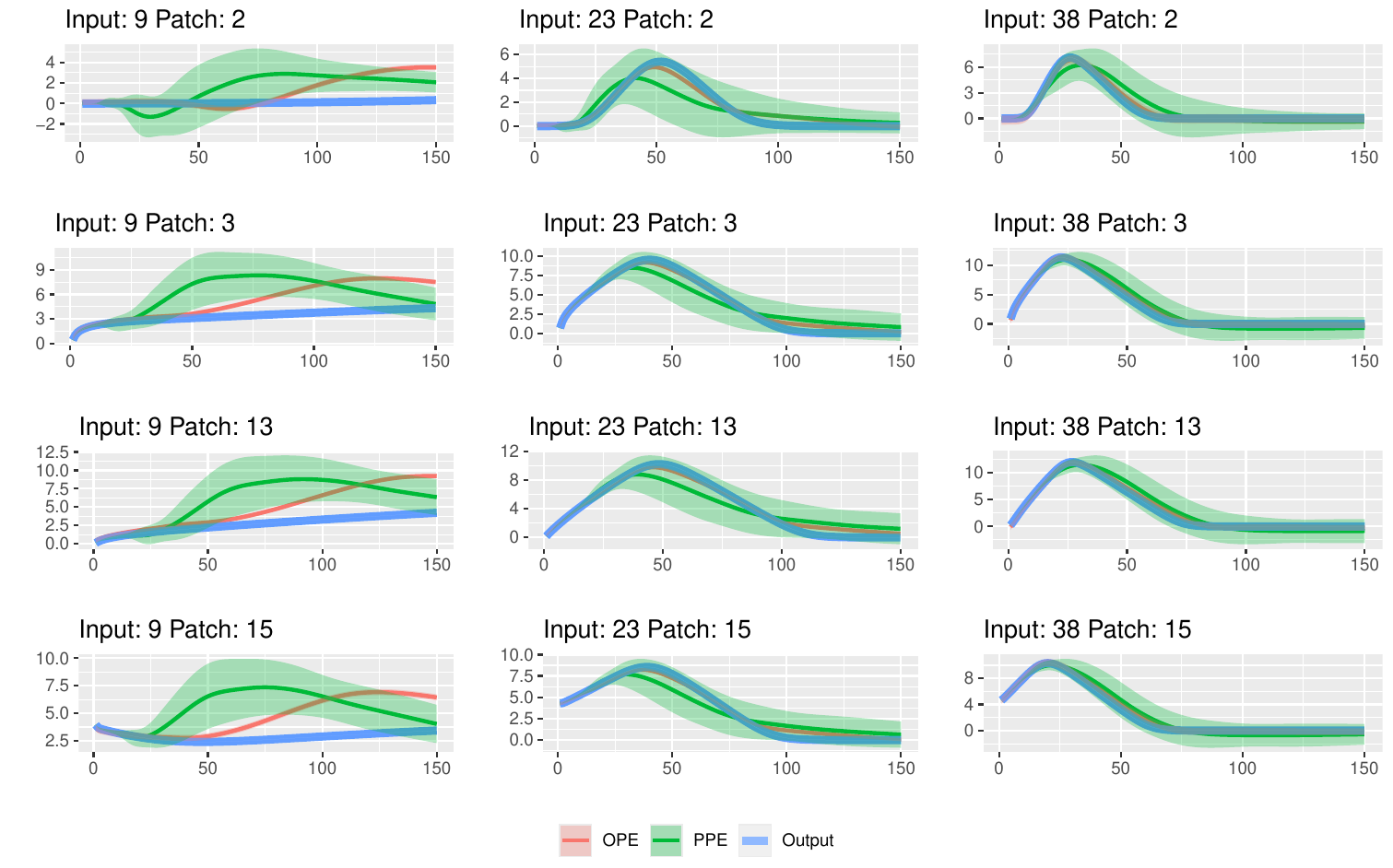}
\caption{Emulator mean predictions (solid lines) and 95\% credible intervals (shaded area) for the influenza simulator for three test input points (columns) and four patches (rows). The OPE (red) and PPE (green) were were built using $n=50$ training points. In each case, true simulator output is given as a blue line for comparison. Predictions are on the log scale.} 
\label{fig:OPEvPPE50}
\end{center}
\end{figure}

Figures~\ref{fig:OPEvPPE20} ($n=20$) and~\ref{fig:OPEvPPE50} ($n=50$) give emulator predictions  for four arbitrarily selected patches and three points from the diagnostic set $\mX_D$. The values of the input parameters for these three diagnostic points are given in Table~\ref{tab:inputs}.

\begin{table}
\caption{Diagnostic input parameters for the influenza simulator corresponding to the points in Figures~\ref{fig:OPEvPPE20} and~\ref{fig:OPEvPPE50}, scaled to $[-1,1]$.}
\label{tab:inputs}
\centering
\begin{tabular}{cccc}
\toprule
Input &  $\alpha$ & $\beta$ & $\gamma$ \\
\midrule
9 & 0.19 & -0.92 & 0.99 \\
23 & 0.68 & -0.73 & 0.20 \\
38 & 0.67 & -0.54 & -0.78 \\
\bottomrule
\end{tabular}
\end{table}

Predictions for input 9 were notably poor for both OPE and PPE emulators, regardless of training design size (left-hand columns of \figref{fig:OPEvPPE20} and \figref{fig:OPEvPPE50}). This input point is near the edge of the input space (Table~\ref{tab:inputs} and Figure~\ref{fig:dist_nearest}), limiting the ability of the emulators to borrow information from nearby points. Although input 9 is actually closer to a design point than either inputs 23 or 38, the number of design points within a sphere of radius $\rho$ is smaller for input 9 than for either of the other inputs for almost all values of $0\le\rho\le 3$, see Figure~\ref{fig:sphere}. Additionally, the specific parameter values associated with input 9 resulted in minimal disease dynamics, with almost no infections in any patch. This type of scenario was rare in the training and test designs and hence poses a challenge for the emulators to capture accurately.

\begin{figure}
\centering
\begin{tabular}{cc}
(a) & (b) \\
\includegraphics[width=.45\linewidth]{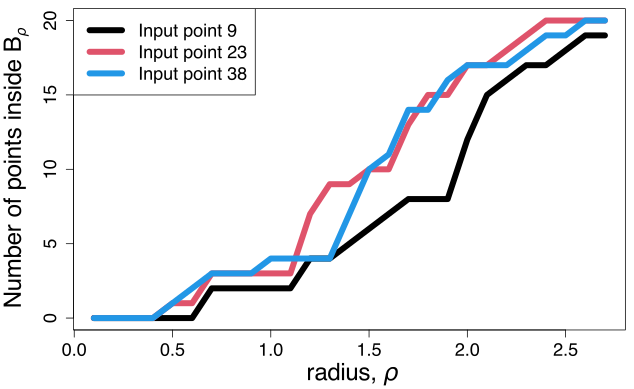} &
\includegraphics[width=.45\linewidth]{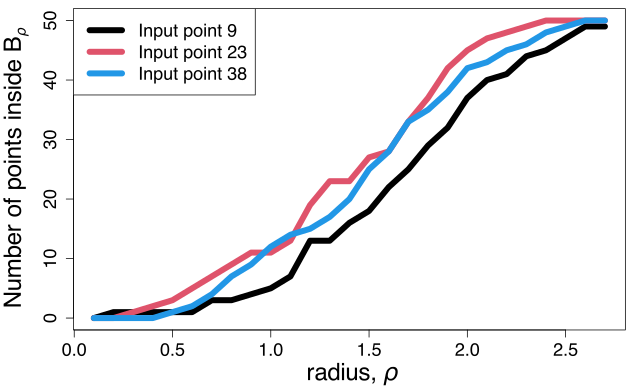}
\end{tabular}
\caption{\label{fig:sphere}Numbers of design points within balls B$_\rho$ of radius $\rho$ from the three diagnostic points for the influenza simulator example for the training designs with (a) $n=20$ and (b) $n=50$.}
\end{figure}

For a more holistic comparison of the emulators, \figref{fig:flu_diag} graphically assesses prediction for $10^4$ input-space-time points (combinations of $\bx^{(i)}$, $\tilde{s}_{i_1}$ or ${s}_{i_1}$ and $t_{i_2}$), randomly selected from the diagnostic set $\mX_D$. Table~\ref{tab:num_diag} gives numerical summaries of the predictive performance of each emulator (with definitions of the summaries in Appendix~\ref{appendix:diag}). Unsurprisingly, both emulators are improved by increasing the training design size from $n=20$ to $n=50$. The differing structure of the models is also apparent in the results; for example, the predictive variance for the OPE is constant across space-time, whereas each location has a potentially different variance for the PPE. This feature leads to more appropriate uncertainty quantification for the PPE, reflected in the mean absolute squared prediction error (MASPE) and mean generalized entropy score (MGES) values in Table~\ref{tab:num_diag}. However, the OPE does appear more accurate (lower root mean squared prediction error (RMSPE)) than the PPE for both $n=20$ and $n=50$, possibly due to the spatial-temporal structure incorporated through both the choice of regressors and space/time correlation matrices. 

Clearly, an advantage of the OPE is the direct and explicit incorporation of the spatio-temporal structure and exploitation of any continuity in these dimensions. This feature was somewhat negated in the influenza example by the limited, and discrete, nature of the spatial information. In the next section, we demonstrate the benefits of the OPE of a metric space in both the space and time dimensions.

\begin{figure} 
 \begin{center}
\begin{subfigure}{.49\textwidth}
    \centering
    \includegraphics[width=\textwidth]{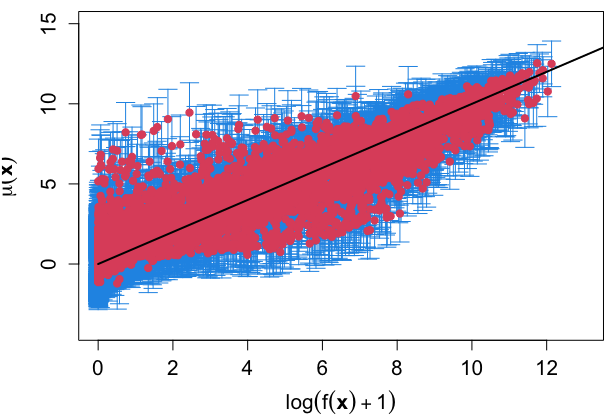}
    \caption{OPE $n = 20$}
\end{subfigure}
\begin{subfigure}{.49\textwidth}
    \centering
    \includegraphics[width=\textwidth]{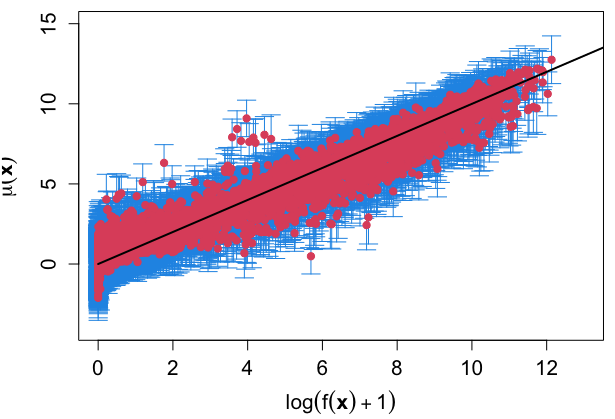}
    \caption{OPE $n = 50$}
\end{subfigure}

\begin{subfigure}{.49\textwidth}
    \centering
    \includegraphics[width=\textwidth]{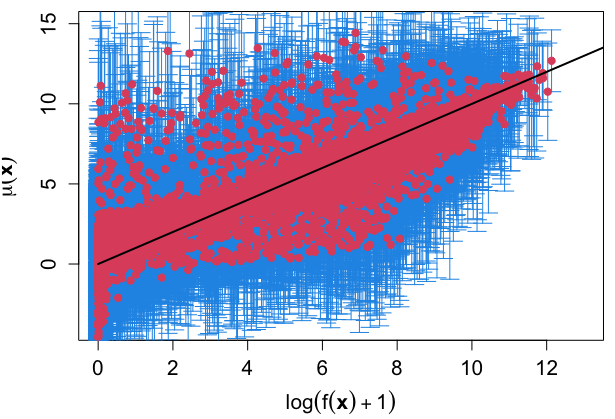}
    \caption{PPE $n = 20$}
\end{subfigure}
\begin{subfigure}{.49\textwidth}
    \centering
    \includegraphics[width=\textwidth]{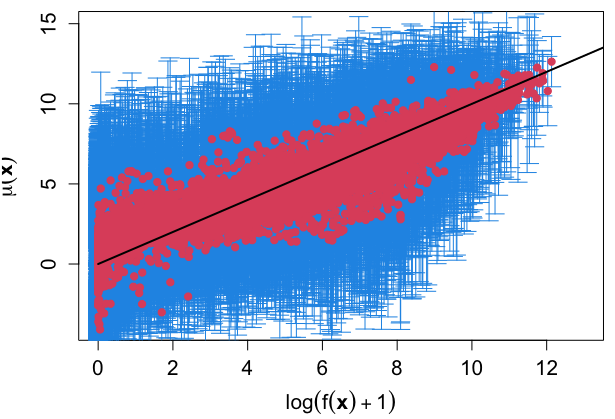}
    \caption{PPE $n = 50$}
\end{subfigure}
\caption{Emulator expectation $\pm3$ standard deviations against influenza simulator output for $10^4$ points (input, space, time) from the set $\mX_D$ of $n^\prime = 150$ diagnostic runs for OPE and PPE with $n=20$ and $n=50$ training points.} 
\label{fig:flu_diag}
\end{center}
\end{figure}

\begin{table}
\caption{The mean absolute squared prediction error (MASPE)\mtnote{1}, root mean squared prediction error (RMSPE)\mtnote{2} and mean generalized entropy score (MGES)\mtnote{3} for each of the four approximations and two design sizes discussed in Sections~\ref{sec:app1} and \ref{sec:app2}. The performance measures are calculated on the test set $\mX_D$. See Appendix~\ref{appendix:diag} for definitions. \label{tab:num_diag}}
\begin{threeparttable}
\begin{tabular}{@{}p{\textwidth}@{}}
\centering
\begin{tabular}{llrrrr}
\toprule
 & & \multicolumn{2}{c}{$n=20$} & \multicolumn{2}{c}{$n=50$} \\ 
Simulator & Diagnostic & OPE & PPE & OPE &  PPE \\
\midrule
& MASPE &  2.758 & 0.736 & 2.142 &  0.664   \\
Influenza & RMSPE &  1.107 & 1.424 & 0.632  & 1.080 \\
& MGES ($\times 10^3$) & -4586 & -270 & -2692 & -37 \\ 
&&&&&\\
& MASPE  & 0.818 & 0.890  & 0.724 & 0.823 \\
Environmental & RMSPE & 0.112 & 0.133 & 0.069 & 0.149 \\
& MGES ($\times 10^3$) & 814 & 1260 & 1047  & 1885 \\
\bottomrule
\end{tabular}
\end{tabular}
\begin{tablenotes}[para]
\item [1] MASPE should be broadly close to $1$.
\item [2] RMSPE is a smaller-the-better quantity.
\item [3] MGES is a larger-the-better quantity
\end{tablenotes}
\end{threeparttable}
\end{table}

\newpage
\section{Synthetic Application: Environmental Simulator}\label{sec:app2}

\subsection{The Simulator}
As a second example, consider an environmental simulator modeling a pollutant spill caused by a chemical accident \cite{bliznyuk2008bayesian}. A mass $x_1\in[7,13]$ of pollutant is spilled at each of two locations in a long, narrow channel, denoted by the space-time vectors $(0, 0)$ and $(x_2, x_3)$. The first spill is fixed; however, the time and spatial location of the second spill are inputs. The location of the second spill is almost the entire domain, $x_2\in[0.01,3]$, but the time of the second spill is quite restrictive, $x_3\in[30.01,30.295]$. A fourth input variable, $x_4\in[0.02,0.12]$, is the constant diffusion rate in the channel.

The simulator output, log pollutant concentration at spatial location $s_{i_1}\in[0,3]$ and time $t_{i_2}\in[0,60]$ is defined as the following closed-form function of $\bx = (x_1, \ldots, x_4)^\top$: 
\begin{align}\label{eq:env}
\begin{split}
    \simx{i_1i_2} & = \log\left\{\sqrt{4\pi}C(\bx;\,s_{i_1}, t_{i_2}) + 1\right\}\,\quad i_1 = 1,\ldots,r_1; \, i_2 = 1, \ldots, r_2\,, \\
    C(\bx;\,s, t) & = \frac{x_1}{\sqrt{4\pi x_4t}}\exp\left(\frac{-s^2}{4x_4t}\right)+\frac{x_1}{\sqrt{4\pi x_4(t-x_3)}}\exp\left(\frac{-(s-x_2)^2}{4x_4(t-x_3)}\right)\indicator(x_3>t)\,,
\end{split}
\end{align}
where  $\indicator$ is the indicator function. Figure~\ref{fig:env} shows the simulator output as a function of $s$ and $t$ for the values $x_1 = 10, x_2 = 1.505,x_3 = 30.1525$ and $x_4 = 0.07$.

Note that the time and location of the first release are fixed and the time of the second release varies slightly around the halfway point of the time interval. Therefore, for $t\in(0,\sim30)$, the simulator output only depends on $x_1$, the mass of pollutant spilled at each location, and $x_4$, the diffusion rate. Hence, emulation of the output in this time region is straightforward. However, for $t\in(\sim30, 60)$, the output is very sensitive to the other inputs, making emulation more difficult, see Figure~\ref{fig:env_trace}. This is especially true (i) around the second release at $t\approx 30$; and (ii) if a good correlation structure is not assumed for the output space.

\begin{figure}[ht]
\begin{subfigure}{.45\textwidth}
    \centering
    \includegraphics[width=\textwidth]{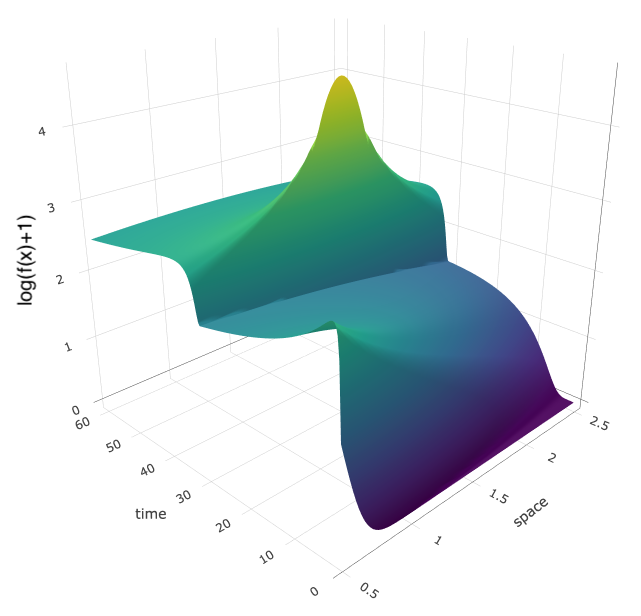}
    \caption{}
    \label{fig:env}  
\end{subfigure}
\begin{subfigure}{.54\textwidth}
    \centering
    \includegraphics[height=7.2cm]{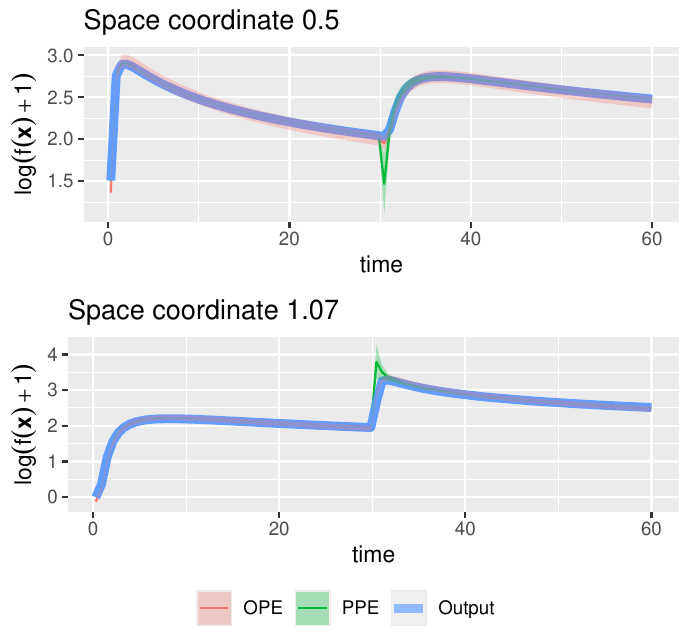}
    \caption{}
     \label{fig:env_trace}  
\end{subfigure}
\caption{Output for the environmental simulator~(\ref{eq:env}) for the input values $x_1 = 10, x_3 = 30.1525, x_2 = 1.505$ and $x_4 = 0.07$. (a) As a function of space and time.  (b) Temporal slices at spatial locations $s = 0.5$ and $s=1.07$ with OPE and PPE mean predictions $\pm 3$ standard deviations from $n=50$ training points.}
\end{figure}

\subsection{Emulation}\label{sec:emulation2}
Once again, tensor-variate GP emulators were constructed under each of the OPE (Section~\ref{sec:OPE}) and PPE (Section~\ref{sec:PPE}) assumptions using designs and prior distributions from Section~\ref{sec:priorest}. For the OPE, the prior mean and correlation structure directly used spatial locations, with $\bg_1^\top(s_{i_1}) = (1, s_{i_1})$ and
$$
(\Sig_1)_{i_1i_1^\prime} = \exp\left(-\{s_{i_1} - s_{i_1^\prime}\}^2/\theta_1\right)\,,
$$
with hyperparameter $\theta_1>0$. Each simulator run produced concentrations for $r_1 = 15$ locations and $r_2 = 100$ times, defined by the discrete sets $s_{i_1}\in\mS = \{0.50, 0.64, \ldots, 2.50\}$ and $t_{i_2}\in\mT = \{0.3, 0.9,\ldots, 60\}$ \cite{bliznyuk2008bayesian}. 

\subsection{Results}\label{sec:results2}

\figref{fig:env_trace} shows examples of two sets of temporal predictions from the OPE and PPE for different spatial locations and one set of input parameters. The key feature revealed by these plots is the difficulty exhibited by PPE in predicting around the time of the second release, with a spike in variance and either under- or over-prediction in the mean. These results are fairly representative of the overall, qualitative, comparison of the two emulators, see \figref{fig:env:trace} for further examples.

As in Section~\ref{sec:app1}, \figref{fig:env_diag} provides an assessment of predictive performance for $10^4$ input-space-time points (combinations of $\bx^{(i)}$, ${s}_{i_1}$ and $t_{i_2}$), randomly selected from the diagnostic set $\mX_D$; numerical summaries (MASPE, RMSPE, MGES) are again provided in Table~\ref{tab:num_diag}. The difficulty of PPE to predict at the second release is evident in the much larger (than OPE) prediction uncertainty for larger simulator outputs, especially around $\simx{i_1 i_2} = 2$. OPE performs better for these difficult predictions, albeit with higher variance for small values of $\simx{i_1 i_2}$ due to the constant space-time variance assumption. This assumption also leads to PPE having larger (better) MGES values in Table~\ref{tab:num_diag}, influenced by more precise predictions in these regions. However, the OPE has substantially lower RSMPE values, especially for $n=50$. Interestingly, increasing the run size from $n=20$ to $n=50$ does little to help with the poor PPE predictions around $\simx{i_1 i_2} = 2$, although it does reduce prediction variance even more for small $\simx{i_1 i_2}$, For OPE, increasing the design size leads to an overall decrease in prediction variance.    

The combination of a metric space in both the space and time locations \cite{rougier2008eef}, combined with spatio-temporal structure in both the regressors and correlation function, provides OPE with an advantage in this example. The concentration output at a given spatial location near the time of the second release can be very variable, depending on the inputs (especially the location of the second release). Hence the PPE near that time point can have high uncertainty, as it does not have the same ability to borrow strength across the space-time locations.

\begin{figure}
\begin{center}
\begin{subfigure}{.49\textwidth}
    \centering
    \includegraphics[width=\textwidth]{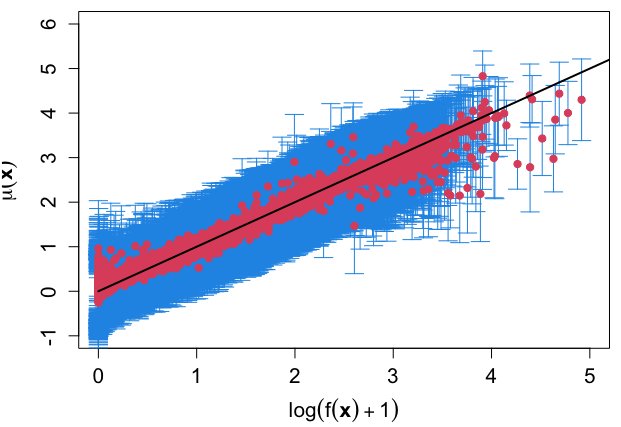}
    \caption{OPE $n = 20$}
\end{subfigure}
\begin{subfigure}{.49\textwidth}
    \centering
    \includegraphics[width=\textwidth]{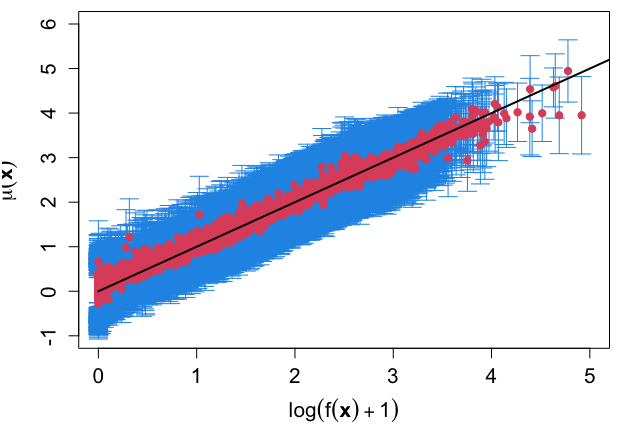}
    \caption{OPE $n = 50$}
\end{subfigure}

\begin{subfigure}{.49\textwidth}
    \centering
    \includegraphics[width=\textwidth]{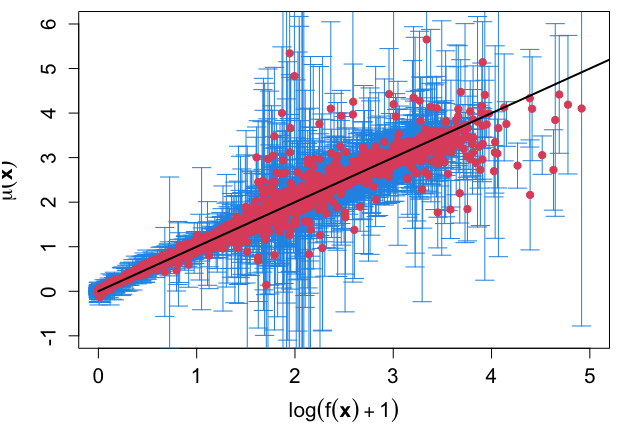}
    \caption{PPE $n = 20$}
\end{subfigure}
\begin{subfigure}{.49\textwidth}
    \centering
    \includegraphics[width=\textwidth]{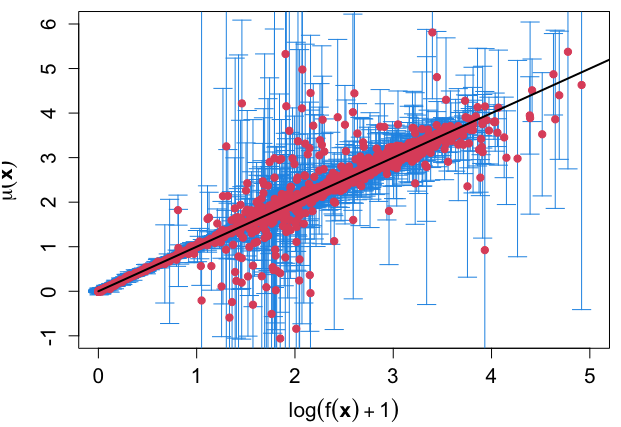}
    \caption{PPE $n= 50$}
\end{subfigure}
\end{center}
\caption{Emulator expectation $\pm3$ standard deviations against environmental simulator output value for $10^4$ points (input, space, time) from the set $\mX_D$ of $n^\prime = 150$ diagnostics runs for OPE and PPE with $n = 20$ and $n = 50$ training points.} 
\label{fig:env_diag}
\end{figure}


\section{Conclusions\label{sec:conc}}

Tensor-variate Gaussian processes are a powerful tool for emulation of multi-output simulators, which generalize previous approaches. Applications of such simulators are common in areas such as biomedicine \cite{sc2020} and engineering \cite{toal2023}, in addition to spatio-temporal modeling. There are also synergies with multi-task learning using Gaussian processes \cite{bcw2007}. 

The important special cases of OPE and PPE have been compared across two case studies, using a practically important simulator from public health and a synthetic environmental simulator. PPE has modeling advantages when the simulator outputs are unstructured, or the structure adds little to the understanding of the system, particularly if larger computer experiments can be performed. When an advantage can be taken of correlation in the output dimensions, OPE can provide more accurate predictions, especially in cases where output structure can compensate for high variation with respect to the simulator inputs. This higher accuracy does come at a higher computational cost, especially for larger numbers of output locations. However, for many applications, computational expense will be dominated by the computer experiment on the simulator, rather than the cost of fitting the emulator.

Extensions of this work could address the additional computation complexity of the TvGP regression, particularly the OPE, through kernel choice for the coregionalization matrices (e.g., tapered kernels \cite{kbhhf2011}) or include heteroskedastic errors across output dimensions. Variable selection could be considered for general TvGP regression and for the OPE in particular, including the choice of regressors in the output locations. TvGP regression could also be applied to provide prediction for output locations that differ from those of the simulator, provided a continuous output domain can be defined.   

\section*{Acknowledgments}
We are grateful to CrystalCast project members for their invaluable comments and discussion. Particular thanks are due to Professor Veronica Bowman and Dr Hannah Williams (both Defence Science and Technology Laboratory, UK), and Dr Thomas Finnie (UK Health Security Agency) for providing the influenza simulator and associated discussions.




\bibliographystyle{unsrt}


\appendix
\counterwithin{lemma}{section}
\numberwithin{equation}{section}
\numberwithin{figure}{section}
\renewcommand{\theequation}{A.\arabic{equation}}
\renewcommand{\thefigure}{A.\arabic{figure}}

\section{Summaries of predictive performance}\label{appendix:diag}

The following quantities are used to summarize emulator predictive performance. In each case, $\mu_f(\cdot)$ is the emulator mean, $v_f(\cdot)$ the emulator variance and $\left\{\tilde{\bx}^{(k^\prime)}, f(\tilde{\bx}^{(k^\prime)})\right\}_{k^\prime=1}^{n^\prime}$ is a set of test points (inputs and simulator outputs). 

\begin{enumerate}
\item Mean Absolute Standardised Prediction Error (MASPE) \cite{jackson2023}:
\be
\frac{1}{n^\prime} \sum_{k^\prime=1}^{n^\prime} \frac{ \left| f(\tilde{\bx}^{(k^\prime)}) - \mu_f(\tilde{\bx}^{(k^\prime)}) \right| }{\sqrt{\nu_f(\tilde{\bx}^{(k^\prime)})}}\,.\label{eq_MASPE}
\ee
\item Root Mean Squared Prediction Error (RMSPE) \cite{bastos2008dfg}:
\be
\sqrt{ \frac{1}{n^\prime} \sum_{k^\prime=1}^{n^\prime} \left\{f(\tilde{\bx}^{(k^\prime)}) - \mu_f(\tilde{\bx}^{(k^\prime)}) \right\}^2 }\,.  
\label{eq_RMSE}
\ee
\item Mean Generalised Entropy Score (MGES) \cite{gneiting2007sps}:
\be
- \, \frac{1}{n^\prime}\sum_{k^\prime=1}^{n^\prime} \left \{ \frac{ \left[ f(\tilde{\bx}^{(k^\prime)}) - \mu_f(\tilde{\bx}^{(k^\prime)}) \right] ^2 }{\nu_f(\tilde{\bx}^{(k^\prime)})} + \log\left(\nu_f(\tilde{\bx}^{(k^\prime)})\right)  \right \}\,. \label{eq_GR27}
\ee
\end{enumerate}

MASPE is a measure of emulator validity; heuristically we expect this value to be around 1 (assuming normal errors this value should be $ \sqrt{2/\pi} $ based on properties of the half-normal distribution).  RMSPE permits comparison of emulator accuracy.  MGES is larger (better) for approximations that are both accurate and valid.

\section{Supporting comparison of the emulators}\label{appendix:comp}
\subsection{Vectorized mean tensors}\label{appendix:vec}
Vectorization of the mean tensor $\ten{M}$, applying definition~(\ref{eq:vec}) with the OPE mean~(\ref{eq:OPEmean}), leads to a mean vector
\be\label{eq:vectOPE}
\vect\left({\ten{M}}\right) = \left(\otimes_{i=0}^m G_i\right)\bb\,,
\ee
with 
$$
G_i = 
\begin{bmatrix}
\mathbf{g}_i^\top(\bx^{(1)}) \\
\vdots \\
\mathbf{g}_i^\top(\bx^{(n)})
\end{bmatrix}\,.
$$
Similarly, assuming the PPE mean~(\ref{eq:PPEmean}) leads to a mean vector
\be\label{eq:vectPPE}
\vect\left({\ten{M}}\right) = \left(G_0\otimes I_{r}\right)\vect\left(\ten{B}\right)\,.
\ee

Hence, when vectorized via~(\ref{eq:vecMN}), both emulators can be cast in the outer product emulator framework; the mean vector takes the form of a linear predictor with both the model matrix and the variance-covariance matrix having Kronecker form and efficient computation can be achieved \cite{rougier2008eef}. In fact, the independence of the parallel emulators in the PPE leads to even more efficient computation.

\subsection{The role of spatio-temporal coregionalization}\label{appendix:coreg}

From \eqref{eq:TvGPpredmean} it is clear that, conditional on the mean tensor $\ten{M}(\tilde{\bx})$, the posterior predictive mean at new input $\tilde{\bx}$ is independent of the coregionalization matrix $\Sig = \otimes_{i=1}^m\Sig_i$. In fact, for the PPE with a non-informative prior on regression parameters, it can be shown that this independence also holds unconditionally \cite{gu2016ppg}. We restate the result here using the notation from this paper.

\begin{lemma}\label{lemma:ppe}
Assume a TvGP implying~(\ref{eq:vecMN}) with vectorized mean tensor~(\ref{eq:vectPPE}) and constant prior on $\vect(B)$. The generalized least squares estimator for $\vect(B)$, or equivalently the posterior mean, is independent of the coregionalization matrix $\Sig = \otimes_{i=1}^m\Sig_i$ and has the form
$$
\widehat{\vect(\ten{B})} = \left[\left(G_0^\top K^{-1}G_0\right)^{-1}G_0^\top K^{-1} \otimes I_r \right]\vect(\ten{F})\,.
$$
\end{lemma}

\begin{proof}
For the most general tensor PPE we have 
$$
\begin{aligned}
\widehat{\vect(\ten{B})} & = \left[\left(G_0^\top \otimes I_r \right)\left(K^{-1}\otimes \Sig^{-1}\right)\left(G_0\otimes I_r \right)\right]^{-1}\left(G_0^\top\otimes I_r \right)\left(K^{-1}\otimes \Sig^{-1}\right)\vect(\ten{F}) \\
& = \left[\left(G_0^\top K^{-1}G_0 \right)^{-1}\otimes \Sig\right]\left(G_0^\top K^{-1}\otimes \Sig^{-1}\right)\vect(\ten{F}) \\
& = \left[\left(G_0^\top K^{-1}G_0\right)^{-1}G_0^\top K^{-1}\otimes I_r \right]\vect(\ten{F})\,,
\end{aligned}
$$
which does not involve $\Sig$.
\end{proof}

The form of the GLS estimator from Lemma~\ref{lemma:ppe} also demonstrates the independence of the estimators for each output location.

In contrast, the GLS estimator for the tensor OPE does not simplify to such an extent, even with non-informative priors on the regression coefficients, as demonstrated via Lemma~\ref{lemma:ope}.

\begin{lemma}\label{lemma:ope}
Assume a TvGP implying~(\ref{eq:vecMN}) with vectorized mean tensor~(\ref{eq:vectOPE}) and constant prior on $\bb$. The generalized least squares estimator for $\bb$, or equivalently the posterior mean, depends on the coregionalization matrix $\Sig = \otimes_{i=1}^m\Sig_i$ and has the form
$$
\hat{\bb} = \left[\left(G_0^\top K^{-1} G_0\right)^{-1}G_0^\top K^{-1} \otimes 
\left(G_{1:m}^\top \Sig^{-1} G_{1:m}\right)^{-1}G_{1:m}^\top\Sig^{-1}\right]\vect(\ten{F})\,,
$$
with $G_{1:m} = \otimes_{i=1}^m G_i$.
\end{lemma}

\begin{proof}
For the most general tensor OPE we have
$$
\begin{aligned}
\hat{\bb} & = \left[\left(G_0^\top \otimes G_{1:m}^\top\right)\left(K^{-1}\otimes \Sig^{-1}\right)\left(G_0 \otimes G_{1:m}\right)\right]^{-1}\left(G_0^\top \otimes G_{1:m}^\top\right)\left(K^{-1}\otimes \Sig^{-1}\right)\vect(\ten{F}) \\
& = \left[\left(G_0^\top K^{-1} G_0\right)^{-1} \otimes \left(G_{1:m}^\top \Sig^{-1} G_{1:m}\right)^{-1}\right]\left( G_0^\top K^{-1} \otimes G_{1:m}^\top\Sig^{-1}\right)\vect(\ten{F}) \\
& = \left[\left(G_0^\top K^{-1} G_0\right)^{-1}G_0^\top K^{-1} \otimes \left(G_{1:m}^\top \Sig^{-1} G_{1:m}\right)^{-1}G_{1:m}^\top\Sig^{-1}\right]\vect(\ten{F})\,.
\end{aligned}
$$
Hence the GLS equations are a Kronecker product of equations across the input and output regressors involving $\Sig$.
\end{proof}
Therefore, for a tensor OPE the mean prediction has the property of smoothing across output locations, which can be beneficial when there are 
strong output dependencies.

\newpage
\clearpage
\section{Additional figures}
Here, we present some additional supporting figures for the two case studies.

\begin{figure}[H]
\begin{subfigure}{\textwidth}
\centering
\includegraphics[width=\textwidth]{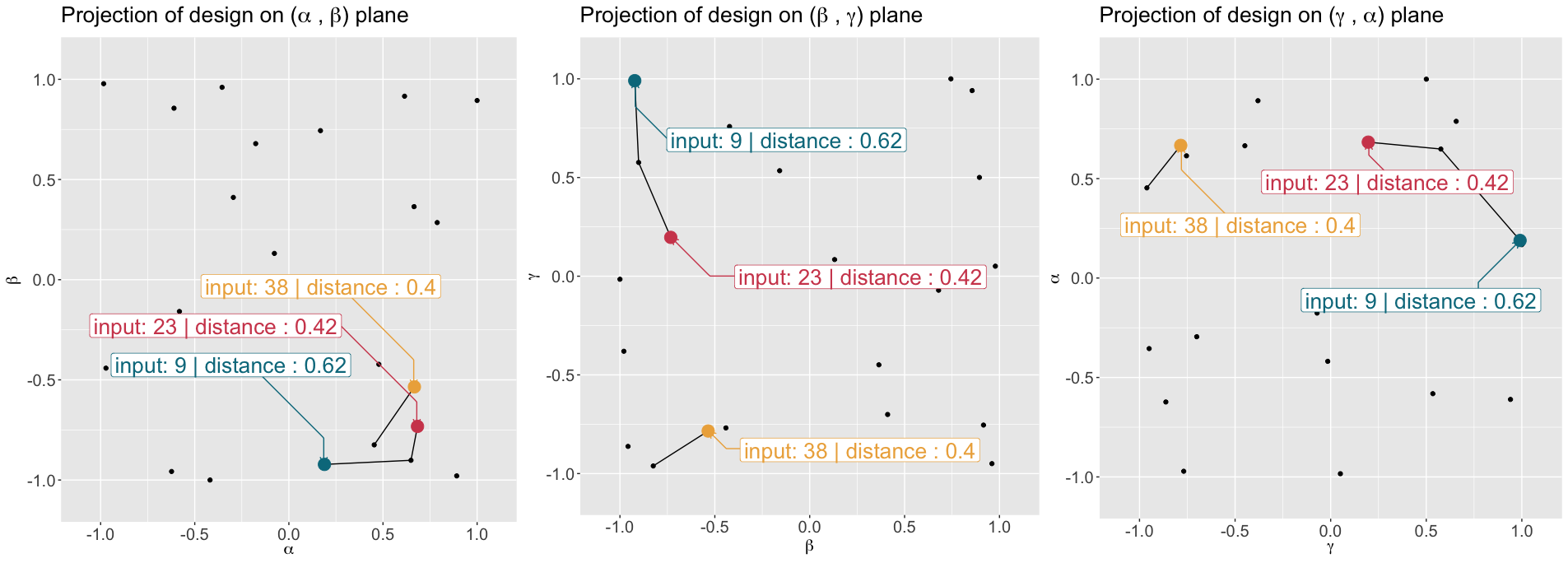}
\caption{Training design size: 20}
\end{subfigure}
\begin{subfigure}{\textwidth}
\centering
\includegraphics[width=\textwidth]{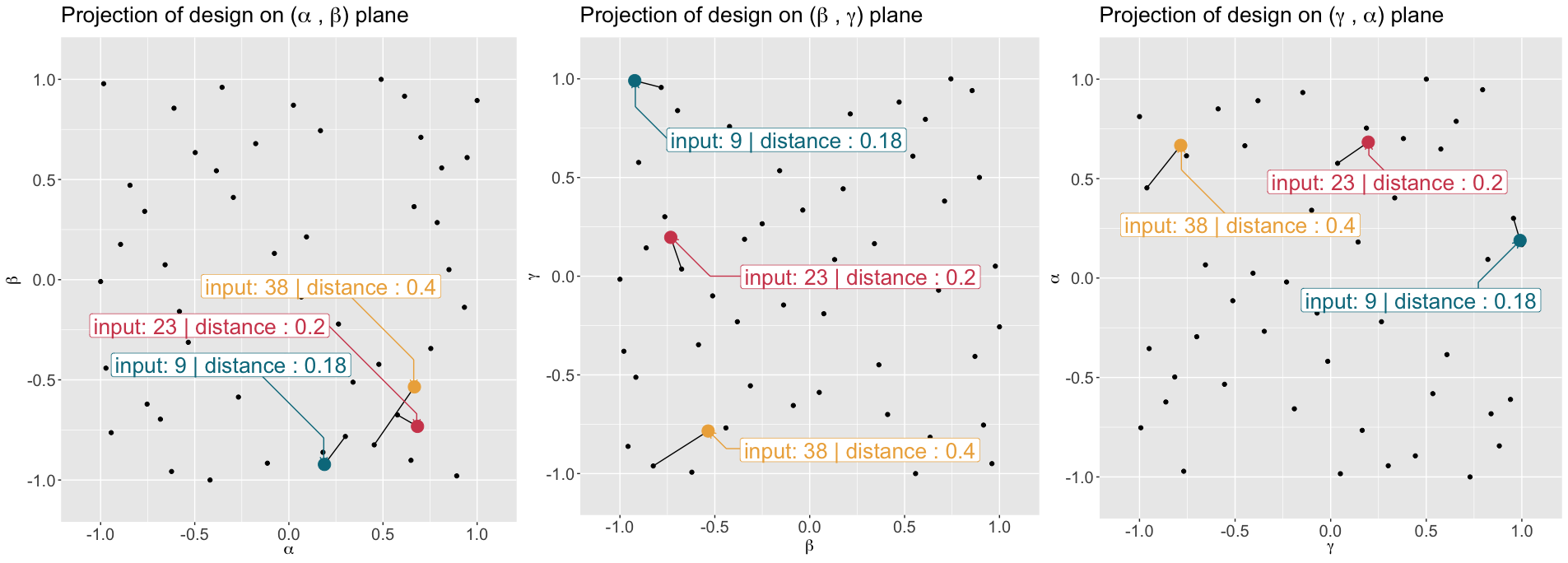}
\caption{Training design size: 50}
\end{subfigure}
\caption{Two-dimensional projections of the inputs for each of the test points presented in Figures~\ref{fig:OPEvPPE20} and~\ref{fig:OPEvPPE50} for the influenza simulator. Training design points are plotted as black dots. Each test point is connected to the nearest design point and the distance is stated in the label.}
\label{fig:dist_nearest}
\end{figure}

\begin{figure}
\centering
\includegraphics[width=.95\textwidth]{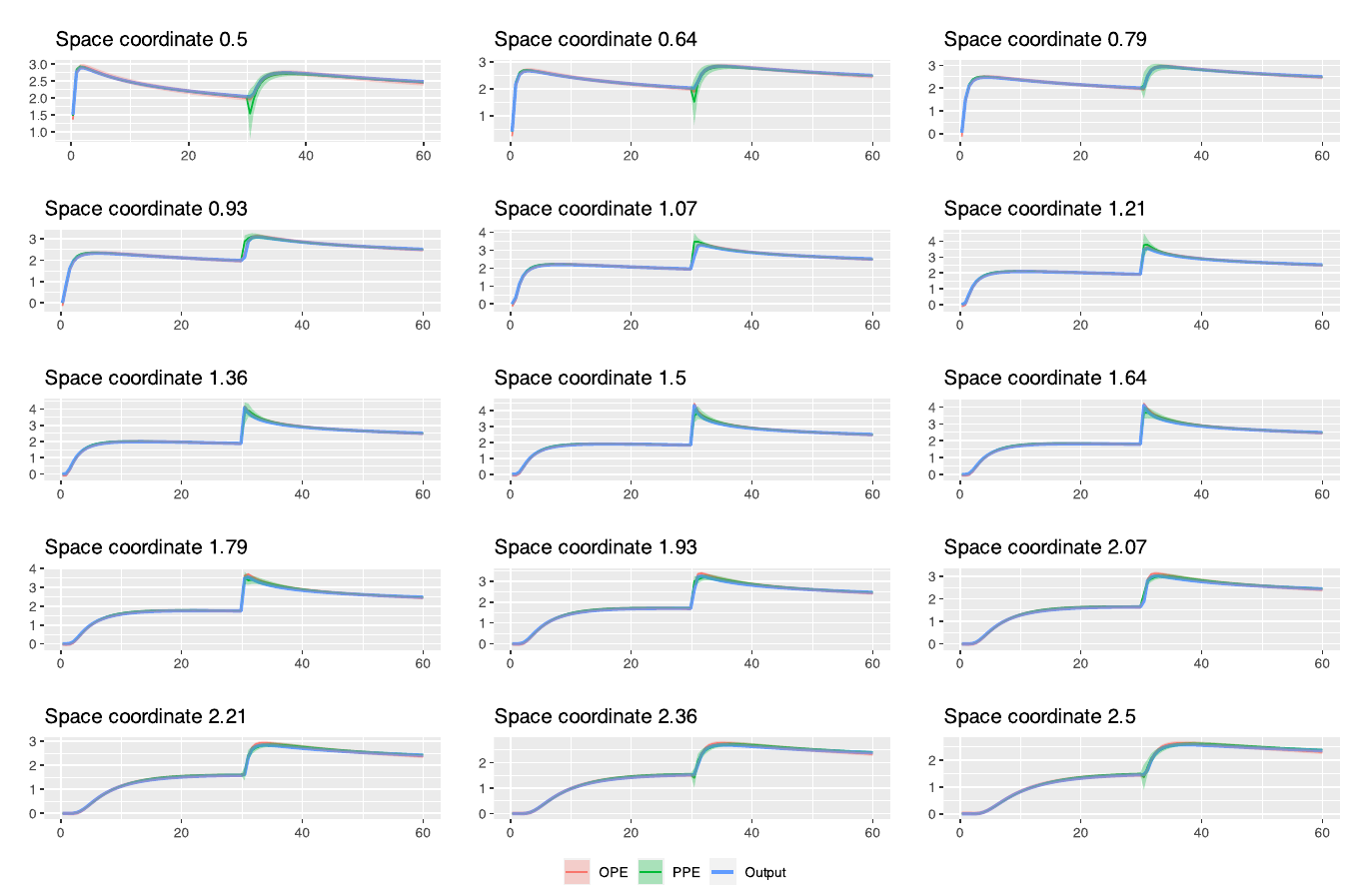}

\caption{\label{fig:env:trace}
Trace plots for the environmental simulator output, OPE and PPE mean predictions and uncertainty ($\pm 3$ standard deviations using $n=50$ training points) for 15 spatial locations at input $x_1 = 10, x_2 = 1.505, x_3 = 30.1525, x_4 = 0.07$.}
\end{figure}

\end{document}

%% file: Definitions.tex
\usepackage{lscape,epsfig,amsmath,setspace}
\usepackage{xypic}

\usepackage{acronym}
\usepackage{enumerate}
\usepackage{color}
\usepackage{graphicx}
\usepackage{bm}
\usepackage{subcaption}
\usepackage[mathscr]{euscript}
\usepackage{mathtools}
\usepackage{threeparttable}
\usepackage{booktabs}
\usepackage{float}
\usepackage{verbatim}

\usepackage{tikz}
\usetikzlibrary{shapes.geometric,arrows,positioning, backgrounds}
\tikzstyle{box} = [rectangle, rounded corners, minimum width=1.5cm, minimum height=0.5cm,text centered, draw=black]

\usepackage{amsmath,amsfonts,amssymb}
\usepackage[all]{xy}
\usepackage{psfrag}

\usepackage{fullpage} 
\usepackage{amsmath,amsfonts,amssymb}
\usepackage{psfrag}
\usepackage{epsfig}
\usepackage[all]{xy}
\usepackage{caption} 

\usepackage{tikz}
\usetikzlibrary{shapes.geometric,arrows,positioning, backgrounds,automata,decorations}
\tikzstyle{box} = [rectangle, rounded corners, minimum width=1.5cm, minimum height=0.5cm,text centered, draw=black]

\usepackage[refpage]{nomencl}

\newcommand{\be}{\begin{equation}}
\newcommand{\ee}{\end{equation}}
\newcommand{\ba}{\begin{eqnarray}}
\newcommand{\ea}{\end{eqnarray}}
\newcommand{\bi}{\begin{itemize}}
\newcommand{\ei}{\end{itemize}}
\newcommand{\bn}{\begin{enumerate}}
\newcommand{\en}{\end{enumerate}}
\newcommand{\bfi}{\begin{figure}}
\newcommand{\efi}{\end{figure}}
\newcommand{\bcen}{\begin{center}}
\newcommand{\ecen}{\end{center}}

\newcommand{\seq}[2]{#1,\dots,#2}

\newcommand{\summ}[3]{\sum_{#1 = #2}^{#3}}

\newcommand{\ex}{\bee(\bx)}
\newcommand{\fx}{\bof(\bx)}

\newcommand{\epr}{\mathbf{e}(\bx')}
\newcommand{\simx}[1]{f_{#1}(\bx)}

\newcommand{\e}[1]{\ensuremath{{\rm E}[#1]}}

\newcommand{\cov}[2]{\ensuremath{{\rm Cov}\left[#1,#2\right]}}

\def \bee { \mathbf{e} }
\def \bof { \mathbf{f} }
\def \bg { \mathbf{g} }

\def \bu { \mathbf{u} }

\def \bx { \mathbf{x} }

\def \br { \mathbf{r} }

\def \bK { K }
\def \bL { L }

\def \bxi { \mathbf{\xi} }

\def \bb { \bm{\beta} }

\def \bmu { \bm{\mu} }

\def \Sig {\Sigma}

\def \bxi { \bm{\xi} }

\newcommand{\mG}{\mathcal{G}}

\newcommand{\mM}{\mathcal{M}}
\newcommand{\mN}{\mathcal{N}}
\newcommand{\mP}{\mathcal{P}}
\newcommand{\mS}{\mathcal{S}}
\newcommand{\mT}{\mathcal{T}}

\newcommand{\mX}{\mathcal{X}}

\newcommand{\mV}{\mathcal{V}}

\def \mbI { \mathbb{ I } }

\def \mbR { \mathbb{ R } }

\def \real { \mbR }

\def \MGP {\mM \mG \mP }
\def \TGP {\mT \mV \mG \mP}

\def \TN {\mT \mN}

\def \indicator { \mbI }

\def \real { \mbR }

\def \vect {\textrm {vec}}

\let\figref=\ref
\renewcommand{\figref}[1]{Figure \ref{#1}}
\let\eqref=\ref
\renewcommand{\eqref}[1]{Equation (\ref{#1})}

\newcommand{\nom}[2]{\nomenclature{$#1$}{#2}}

\newcommand{\ten}[1]{\mathscr{#1}}

\newcommand{\mtnote}[1]{\textsuperscript{\TPTtagStyle{#1}}}
\newcommand*{\email}[1]{\href{mailto:#1}{\nolinkurl{#1}} } 


%% file: TvGP.bbl
\begin{thebibliography}{10}

\bibitem{andrianakis2012teo}
Y.~Andrianakis and P.~G. Challenor.
\newblock The effect of the nugget on {G}aussian process emulators of computer
  models.
\newblock {\em Computational Statistics and Data Analysis}, 56:4215--4228,
  2012.

\bibitem{bastos2008dfg}
T.~S. Bastos and A.~O'Hagan.
\newblock Diagnostics for {G}aussian process emulators.
\newblock {\em Technometrics}, 51:425--438, 2008.

\bibitem{bayarri2007aff}
J.~M. Bayarri, J.~O. Berger, R.~Paulo, J.~Sacks, J.~A. Cafeo, J.~Cavendish,
  C.~Lin, and J.~Tu.
\newblock A framework for validation of computer models.
\newblock {\em Technometrics}, 49:138--154, 2007.

\bibitem{bzkl2013}
I.~Bilionis, N.~Zabaras, B.~A. Konomi, and G.~Lin.
\newblock Multi-output separable gaussian process: Towards an efficient, fully
  bayesian paradigm for uncertainty quantification.
\newblock {\em Journal of Computational Physics}, 241:212--239, 2013.

\bibitem{bliznyuk2008bayesian}
N.~Bliznyuk, D.~Ruppert, C.~Shoemaker, R.~Regis, S.~Wild, and P.~Mugunthan.
\newblock Bayesian calibration and uncertainty analysis for computationally
  expensive models using optimization and radial basis function approximation.
\newblock {\em Journal of Computational and Graphical Statistics}, 17:270--294,
  2008.

\bibitem{bcw2007}
E.~V. Bonilla, K.~M.~A. Chai, and C.~K.~I. Williams.
\newblock Multi-task {G}aussian process prediction.
\newblock In {\em Advances in Neural Information Processing Systems},
  volume~20. Curran Associates, Inc., 2007.

\bibitem{bowman2016eom}
V.~E. Bowman and D.~C. Woods.
\newblock Emulation of multivariate simulators using thin-plate splines with
  application to atmospheric dispersion.
\newblock {\em SIAM/ASA Journal of Uncertainty Quantification}, 4:1323--1344,
  2016.

\bibitem{byrd1995limited}
R.~H. Byrd, P.~Lu, J.~Nocedal, and C.~Zhu.
\newblock A limited memory algorithm for bound constrained optimization.
\newblock {\em SIAM Journal on Scientific Computing}, 16:1190--1208, 1995.

\bibitem{calder2018}
M.~Calder, C.~Craig, D.~Culley, R.~{de}~Cani, C.~A. Donnelly, R.~Douglas,
  B.~Edmonds, J.~Gascoigne, N.~Gilbert, C.~Hargrove, D.~Hinds, D.~C. Lane,
  D.~Mitchell, G.~Pavey, D.~Robertson, B.~Rosewell, S.~Sherwin, M.~Walport, and
  A.~Wilson.
\newblock Computational modelling for decision-making: where, why, what and
  how.
\newblock {\em Royal Society Open Science}, 5:172096, 2018.

\bibitem{campbell2020}
A.~Campbell and P.~Pietro~Li{\`o}.
\newblock {tvGP-VAE}: Tensor-variate {G}aussian process prior variational
  autoencoder.
\newblock {\em arXiv:2006.04788}, 2020.

\bibitem{chu2009}
W~Chu and Z.~Ghahramani.
\newblock Probabilistic models for incomplete multi-dimensional arrays.
\newblock In D.~van Dyk and M.~Welling, editors, {\em Proceedings of the
  Twelfth International Conference on Artificial Intelligence and Statistics},
  pages 89--96, Clearwater Beach, Florida, 2009. Proceedings of Machine
  Learning Research.

\bibitem{conti2010beo}
S.~Conti and A.~O'Hagan.
\newblock Bayesian emulation of complex multi-output and dynamic computer
  models.
\newblock {\em Journal of Statistical Planning and Inference}, 140:640--651,
  2010.

\bibitem{farah2014bea}
M.~Farah, P.~Birrell, S.~Conti, and D.~De~Angelis.
\newblock Bayesian emulation and calibration of a dynamic epidemic model for
  {A/H1N1} influenza.
\newblock {\em Journal of the American Statistical Association},
  109:1398--1411, 2014.

\bibitem{Gao_2024}
Y.~Gao and E.~B. Pitman.
\newblock Parallel partial emulation in applications.
\newblock {\em International Journal for Uncertainty Quantification}, 14:1--15,
  2024.

\bibitem{gneiting2007sps}
T.~Gneiting and A.~E. Raftery.
\newblock Stricitly proper scoring rules, prediction and estimation.
\newblock {\em Journal of the American Statistical Association}, 102:359--378,
  2007.

\bibitem{goldstein2004pff}
M.~Goldstein and J.~C. Rougier.
\newblock Probabilistic formulations for transferring inferences from
  mathematical models to physical systems.
\newblock {\em SIAM Journal on Scientific Computing}, 26:467--487, 2004.

\bibitem{gramacy2020surrogates}
Robert~B. Gramacy.
\newblock {\em Surrogates: {G}aussian Process Modeling, Design and \
  Optimization for the Applied Sciences}.
\newblock Chapman Hall/CRC, Boca Raton, Florida, 2020.
\newblock \url{http://bobby.gramacy.com/surrogates/}.

\bibitem{gu2016ppg}
M.~Gu and J.~O. Berger.
\newblock Parallel partial {G}aussian process emulation for computer models
  with massive output.
\newblock {\em Annals of Applied Statistics}, 10:1317--1347, 2016.

\bibitem{hethcote2000}
H.~W. Hethcote.
\newblock The mathematics of infectious diseases.
\newblock {\em {SIAM} Review}, 42:599--653, 2000.

\bibitem{higdon2008cmc}
D.~Higdon, J.~Gattiker, B.~Williams, and M.~Rightley.
\newblock Computer model calibration using high-dimensional output.
\newblock {\em Journal of the American Statistical Association}, 103:570--583,
  2008.

\bibitem{toal2023}
Toal. D.~J. J.
\newblock Applications of multi-fidelity multi-output {K}riging to engineering
  design optimization.
\newblock {\em Structural and Multidisciplinary Optimization}, 66:125, 2023.

\bibitem{jackson2018dop}
S.~E. Jackson.
\newblock {\em Design of Physical System Experiments Using Bayes Linear
  Emulation and History Matching Methodology with Application to Arabidopsis
  Thaliana}.
\newblock PhD thesis, Durham University, 2018.

\bibitem{jackson2023}
S.~E. Jackson and I.~Vernon.
\newblock Efficient emulation of computer models utilising multiple known
  boundaries of differing dimension.
\newblock {\em Bayesian Analysis}, 18:165--191, 2023.

\bibitem{jackson2020uhc}
S.~E. Jackson, I.~Vernon, J.~Liu, and K.~Lindsey.
\newblock Understanding hormonal crosstalk in arabidopsis root development via
  emulation and history matching.
\newblock {\em Statistical Approaches in Genetics and Molecular Biology}, 19,
  2020.

\bibitem{joseph2015mpd}
V.~R. Joseph, E.~Gul, and S.~Ba.
\newblock Maximum projection designs for computer experiments.
\newblock {\em Biometrika}, 102:371--380, 2015.

\bibitem{kbhhf2011}
C.~G. Kaufman, D.~Bingham, S.~Habib, K.~Heitmann, and J.~A. Frieman.
\newblock {Efficient emulators of computer experiments using compactly
  supported correlation functions, with an application to cosmology}.
\newblock {\em The Annals of Applied Statistics}, 5:2470 -- 2492, 2011.

\bibitem{kennedy2001bco}
M.~C. Kennedy and A.~O'Hagan.
\newblock Bayesian calibration of computer models (with discussion).
\newblock {\em Journal of the Royal Statistical Society B}, 63:425--464, 2001.

\bibitem{kitnya2023}
A.~P. Kyprioti, C.~Irwin, A.~A. Taflanidis, N.~C. Nadal-Caraballo, M.~C. Yawn,
  and L.~A. Aucoin.
\newblock Spatio-temporal storm surge emulation using gaussian process
  techniques.
\newblock {\em Coastal Engineering}, 180:104231, 2023.

\bibitem{ohagan1999uaa}
A.~O'Hagan, M.~C. Kennedy, and J.~E. Oakley.
\newblock Uncertainty analysis and other inference tools for complex computer
  codes.
\newblock In J.~M. Bernardo, J.~O. Berger, A.~P. Dawid, and A.~F.~M. Smith,
  editors, {\em Bayesian Statistics 6}, pages 503--524. Oxford University
  Press, 1999.

\bibitem{ohlson2013}
M.~Ohlson, M.~{Rauf Ahmad}, and D.~{von Rosen}.
\newblock The multilinear normal distribution: introduction and some basic
  properties.
\newblock {\em Journal of Multivariate Analysis}, 113:37--47, 2013.

\bibitem{overstall2016meo}
A.~M. Overstall and D.~C. Woods.
\newblock Multivariate emulation of computer simulators: model selection and
  diagnostics with application to a humanitarian relief model.
\newblock {\em Journal of the Royal Statistical Society C}, 65, 2016.

\bibitem{rasmussen2006gpf}
C.~E. Rasmussen and C.~K.~I. Williams.
\newblock {\em Gaussian Processes for Machine Learning}.
\newblock MIT Press, Cambridge, Massachusetts, 2006.

\bibitem{rgmr2009}
J.~Rougier, S.~Guillas, A.~Maute, and A.~D. Richmond.
\newblock Expert knowledge and multivariate emulation: The
  thermosphere-ionosphere electrodynamics general circulation model
  {(TIE-GCM)}.
\newblock {\em Technometrics}, 51:414--424, 2009.

\bibitem{rougier2008eef}
J.~C. Rougier.
\newblock Efficient emulators for multivariate deterministic functions.
\newblock {\em Journal of Computational and Graphical Statistics}, 17:827--843,
  2008.

\bibitem{santner2003tda}
T.~J. Santner, B.~J. Williams, and W.~I. Notz.
\newblock {\em The Design and Analysis of Computer Experiments}.
\newblock Springer, New York, 2nd edition, 2018.

\bibitem{sd1995}
L.~Sattenspiel and K.~Dietz.
\newblock A structured epidemic model incorporating geographic mobility among
  regions.
\newblock {\em Mathematical Biosciences}, 128:71--91, 1995.

\bibitem{sbbcppw2014}
E.~T. Spiller, M.~J. Bayarri, J.~O. Berger, E.~S. Calder, A.~K. Patra, E.~B.
  Pitman, and R.~L. Wolpert.
\newblock Automating emulator construction for geophysical hazard maps.
\newblock {\em {SIAM/ASA} Journal of Uncertainty Quantification}, 2:126--152,
  2014.

\bibitem{sc2020}
P.~Stolfi and F.~Castiglione.
\newblock Emulation of dynamic multi-output simulator of risk of type-2
  diabetes.
\newblock In {\em 2020 IEEE International Conference on Bioinformatics and
  Biomedicine}, pages 340--345, 2020.

\bibitem{vdDriessche}
P.~{van den}~Driessche.
\newblock Spatial structure: patch models.
\newblock In {\em Mathematical Epidemiology}, chapter~7, pages 179--189.
  Springer-Verlag, Berlin, 2008.

\bibitem{xu2012}
Z.~Xu, F.~Yan, and Y.~Qi.
\newblock Infinite {T}ucker decomposition: nonparametric {B}ayesian models for
  multiway data analysis.
\newblock In {\em Proceedings of the 29th International Coference on
  International Conference on Machine Learning}, pages 1675–--1682, Madison,
  Wisconsin, 2012. Omnipress.

\bibitem{zhao2014}
Q.~Zhao, G.~Zhou, L.~Zhang, and A.~Cichocki.
\newblock Tensor-variate {G}aussian processes regression and its application to
  video surveillance.
\newblock In {\em 2014 IEEE International Conference on Acoustics, Speech and
  Signal Processing}, pages 1265--1269, 2014.

\end{thebibliography}
